\documentclass[letterpaper, 10 pt, conference]{ieeeconf}
\IEEEoverridecommandlockouts 
\usepackage{amssymb,amsmath,graphicx,float}
\usepackage{algorithm, algpseudocode}
\usepackage{booktabs}
\usepackage{wrapfig}
\usepackage{setspace}
\usepackage{graphicx}
\usepackage{epsfig}            
\usepackage[all]{xy}                             
\usepackage{verbatim}
\usepackage{float}
\usepackage{mathrsfs}  
\usepackage{bm}
\usepackage{graphicx}
\usepackage{float}
\usepackage{booktabs}
\usepackage{algorithm, algpseudocode}
\newtheorem{theorem}{Theorem} 

\newtheorem{definition}{Definition}
 
\newtheorem{proposition}{Proposition}

\newcommand{\norm}[1]{\left\lVert#1\right\rVert}

\newcommand{\argmin}{\mathrm{arg}\min}

\usepackage{fixltx2e}
\usepackage{tabularx}
\usepackage{verbatim}
\usepackage{color}
\usepackage{hyperref}
\usepackage{xmpmulti}
\usepackage{transparent}
\usepackage{fancyhdr}

\begin{document}
	
\title{Provably Correct Controller Synthesis of Switched Stochastic Systems with Metric Temporal Logic Specifications: A Case Study on Power Systems}

\author{Zhe~Xu\thanks{Zhe~Xu is with the School for Engineering of Matter, Transport, and Energy, Arizona State University, Tempe, AZ 85287, USA. Email: {\tt\small $\{$xzhe1@asu.edu$\}$. }}, Yichen~Zhang\thanks{Yichen~Zhang is with Argonne National Laboratory, Lemont, IL 60439, USA. Email: {\tt\small $\{$yichen.zhang@anl.gov$\}$. }} }

\maketitle 

\begin{abstract}
In this paper, we present a provably correct controller synthesis approach for switched stochastic control systems with \textit{metric temporal logic} (MTL) specifications with provable probabilistic guarantees. We first present the \textit{stochastic control bisimulation function} for switched stochastic control systems, which bounds the trajectory divergence between the switched stochastic control system and its nominal deterministic control system in a probabilistic fashion. We then develop a method to compute optimal control inputs by solving an optimization problem for the nominal trajectory of the deterministic control system with robustness against initial state variations and stochastic uncertainties. We implement our robust stochastic controller synthesis approach on both a four-bus power system and a nine-bus power system under generation loss disturbances, with MTL specifications expressing requirements for the grid frequency deviations, wind turbine generator rotor speed variations and the power flow constraints at different power lines. 			
\end{abstract}  

\IEEEpeerreviewmaketitle
\section{Introduction} \label{Introduction}
	

A \textit{switched stochastic system} \cite{liberzon2003switching,Xiang2012} consists of a set of stochastic dynamic \textit{modes} and switchings between the modes triggered by external events. Many cyber-physical systems (e.g., power systems) can be modeled as switched stochastic systems and the control synthesis of such systems with \textit{formal} specifications has been a challenging problem.

In this paper, we present a provably correct controller synthesis approach for switched stochastic control systems with \textit{metric temporal logic} (MTL) specifications. MTL has been used as specifications in power systems \cite{zhe_control}, artificial intelligence \cite{zhe_ijcai2019}, robotics \cite{Verginis2019Icra}, biology \cite{Xu2021PLOS}, etc. We first present the \textit{stochastic control bisimulation function} for switched stochastic control systems, which bounds the trajectory divergence between the switched stochastic control system and its nominal deterministic control system in a probabilistic fashion. Thus all the controller synthesis methods for the nominal deterministic system can be used for designing the optimal input signals, and the same input signals can be applied to the switched stochastic control system with a lower-bound guarantee for satisfying the MTL specifications. 

In \cite{zheACC2018wind}, we presented a coordinated control method of wind turbine generator and energy storage system for frequency regulation under temporal logic
specifications. In this paper, we extend the results in \cite{zheACC2018wind} to switched stochastic control systems, and generalize the predicates of the MTL specifications to include both the state and the input (e.g., so that line power constraints in power systems can be incorporated into the MTL specifications). Besides, we add an exponential term to the stochastic control bisimulation function so that both stable and unstable linear dynamics can be approached with less conservativeness. 

We implement the proposed controller synthesis approach in two scenarios in power systems. The results show that the synthesized control inputs can indeed lead to satisfaction of the MTL specifications with larger empirical probabilities than 
the derived theoretical lower-bound guarantees for the satisfaction probability.

\section{Related Works} 	
There is a rich literature on controller synthesis with temporal logic specifications in the stochastic environment
\cite{Nonlinear2017,Wolff2012}. For discrete-time temporal logics such as co-safe linear temporal logics (LTL), the specifications can be converted to finite state machines, then the optimal control strategy is computed in the state space augmented with the states of the constructed finite state machines \cite{Horowitz2014}. For dense-time temporal logics such as MTL or signal temporal logics (STL), the specifications can be converted to timed automata \cite{AlurDill1990,Fu2015CDC} and the optimal control strategy is computed in the state space augmented with the states of the constructed timed automata. In \cite{Anand2019}, the authors proposed a verification approach of switched stochastic systems with LTL specifications. However, as far as we know, there has been no work on controller synthesis of switched stochastic systems with (dense-time) temporal logic specifications.

\section{Preliminaries} 
	\label{Preliminaries}  
\subsection{Switched Stochastic Control Systems}
\label{switch} 
\begin{definition}[Switched Stochastic Control Systems]
A \textbf{switched stochastic control system} is a 6-tuple $\mathcal{T} = (\mathcal{Q},\mathcal{X},\mathcal{X}_0,\mathcal{V},\mathcal{F},\mathcal{E})$ where
\begin{itemize}
	\item $\mathcal{Q}=\{1,2,\dots,M\}$ is the
	set of indices for the modes (or subsystems);
	\item $\mathcal{X}$ is the domain of the continuous state, $x\in \mathcal{X}$ is the continuous state of the system, $\mathcal{X}_0\subset\mathcal{X}$ is the initial set of states;	
	\item $\mathcal{V}$ is the domain of the input, $u\in\mathcal{V}$ is the input of the system;		
	\item $\mathcal{F}=\{(f_{q},g_{q})\vert q\in \mathcal{Q}\}$ where $f_{q}$ describes the continuous
	time-invariant dynamics for the mode $dx=f_{q}(x,u)dt+g_{q}(x,u)dw$, which
	admits a unique solution $\xi_{q}(t;x_{q}^0,u)$, where $\xi_{q
	}$ satisfies $d\xi_{q}(t;x_{q}^0,u)= f_{q}(\xi_{q}(t;x_{q}^0,u),u)dt+g_{q}(\xi_{q}(t;x_{q}^0,u),u)dw$, and $\xi_{q}(0;x_{q}^0,u)=x_{q}^0$ is an initial condition in mode $q$;
	\item $\mathcal{E}$ is a subset of $\mathcal{Q}\times \mathcal{Q}$
	which contains the valid transitions. If a transition $e = (q, q')\in\mathcal{E}$ takes place, the system switches from mode $q$ to $q'$.
\end{itemize}
\label{sw}
\end{definition}                      
Similarly, we can define the \textbf{switched nominal control system} $\mathcal{T}^{\ast} = (\mathcal{Q},\mathcal{X},\mathcal{X}_0,\mathcal{V},\mathcal{F}^{\ast},\mathcal{E})$ of $\mathcal{T}$, and $\mathcal{T}^{\ast}$  only differs from $\mathcal{T}$ as $\mathcal{F}^{\ast}=\{f_{q}\vert q\in \mathcal{Q}\}$, where $dx^{\ast}=f_{q}(x^{\ast},u)dt$ is the nominal deterministic version of $dx=f_{q}(x,u)dt+g_{q}(x,u)dw$.

  
\begin{definition}[Trajectory]
	\label{def_traj}	
	A trajectory of a stochastic switched control system $\mathcal{T}$ is denoted as a sequence $\rho=\{(q^{i},\xi_{q^{i}}(t;x^{0}_{q^i},u),T^{i})\}_{i=0}^{N_{q}}$ ($N_{q}\in\mathbb{N}$),	
	where  
		\begin{itemize}  
		\item $\forall i\ge0$, $q^{i}\in \mathcal{Q}$, $x^{0}_{q^i}\in\mathcal{X}$ is the initial state at mode $q^i$, $x_0=x^{0}_{q^0}\in \mathcal{X}_0$ is the initial state of the entire trajectory, $x^{i+1}=\xi_{q^{i}}(T^{i};x^{0}_{q^i},u)$ is the initial state at mode $q^{i+1}$;
		\item $\forall i\ge0$, $T^{i}>0$ is the dwell time at mode $q^{i}$, while the transition times are $T^0, T^0+T^1, \dots,T^0+T^1+\dots+T^{N_q-1}$;
		\item  $\forall i\ge0$, $(q^{i},q^{i+1})\in \mathcal{E}$.
		\end{itemize} 
\end{definition} 
A trajectory of a switched nominal control system $\mathcal{T}^{\ast}$ can be similarly denoted as a sequence $\rho^{\ast}=\{(q^{i},\xi^{\ast}_{q^{i}}(t;x^{\ast0}_{q^i},u),T^{i})\}_{i=0}^{N_{q}}$ ($N_{q}\in\mathbb{N}$).
	
\begin{definition}[Output Trajectory]
	\label{traj}
For a trajectory $\rho=\{(q^{i},\xi_{q^{i}}(t;x^{0}_{q^i},u),T^{i})\}_{i=0}^{N_{q}}$, we define the output trajectory $s_\rho(\cdot;x_0,u)$ (here we denote $x_0\triangleq x^{0}_{q^0}$ for brevity) as follows:
\[
 	s_\rho(t;x_0,u) =\begin{cases}
 	\xi_{q^0}(t;x_{0},u), ~~~~~~~~~~~~~~~~~~~~~~~~\mbox{if $t<T^0$},\\
 	\xi_{q^i}(t-\displaystyle{\sum_{k=0}^{i-1}}T^k;x^{0}_{q^i},u),\\
 	~~~~~~~~~~\mbox{if $\displaystyle{\sum_{k=0}^{i-1}}T^k\le t<\sum_{k=0}^{i}T^k$, $1\le i\le N_{q}$}.
 	\end{cases} \\
\]                         
\end{definition} 

The output trajectory of a trajectory $\rho^{\ast}=\{(q^{i},\xi^{\ast}_{q^{i}}(t;x^{\ast0}_{q^i},u),T^{i})\}_{i=0}^{N_{q}}$ of a switched nominal control system is denoted as $s_{\rho^{\ast}}(\cdot;x^\ast_0,u)$.
	
\subsection{Metric Temporal Logic (MTL)}
	\label{MTL}   
In this subsection, we briefly review the MTL that are interpreted over continuous-time signals~\cite{FAINEKOScontinous}.
The domain of the continuous state $x$ is denoted by $\mathcal{X}$. The domain $\mathbb{B} = \{$True, False$\}$ is the Boolean domain and the time set
is $\mathbb{T} = \mathbb{R}$. The output trajectory $s_\rho(\cdot;x_0,u)$ of a switched system is defined in Sec. \ref{switch}. A set $AP=\{\pi_1,\pi_2,\dots \pi_n\}$ is a set of atomic propositions, each mapping $\mathcal{X}$ to $\mathbb{B}$. The syntax of MTL is defined recursively as follows: 
    \[
	\varphi:=\top\mid \pi\mid\lnot\varphi\mid\varphi_{1}\wedge\varphi_{2}\mid\varphi_{1}\vee
	\varphi_{2}\mid\varphi_{1}\mathcal{U}_{\mathcal{I}}\varphi_{2},   
	\]                           
where $\top$ stands for the Boolean constant True, $\pi$ is an atomic
proposition, $\lnot$ (negation), $\wedge$ (conjunction), $\vee$ (disjunction)
are standard Boolean connectives, $\mathcal{U}_{\mathcal{I}}$ is a temporal operator
representing \textquotedblleft until\textquotedblright, $\mathcal{I}$ is a time interval of
the form $\mathcal{I}=[i_{1},i_{2}]~(i_{1},i_{2}\in \mathbb{R}_{\geqslant 0}, i_{1}\le i_{2})$. From \textquotedblleft
until\textquotedblright ($\mathcal{U}_{\mathcal{I}}$), we can derive the temporal operators \textquotedblleft
eventually\textquotedblright~$\Diamond_{\mathcal{I}}\varphi=\top\mathcal{U}_{\mathcal{I}}\varphi$ and
	\textquotedblleft always\textquotedblright~$\Box_{\mathcal{I}}\varphi=\lnot\Diamond_{\mathcal{I}}\lnot\varphi$.
	
	We define the set of states that satisfy the atomic proposition $\pi$ as $\mathcal{O}(\pi)\subset \mathcal{X}$. For a set $S\subseteq\mathcal{X}$, we define the signed distance from $x$ to $S$ as
	\begin{equation}
	\textbf{Dist$_d(x,S)\triangleq$}
	\begin{cases}
	-\textrm{inf}\{d(x, y)\vert y\in cl(S)\},& \mbox{if $x$ $\not\in S$};\\  
	\textrm{inf}\{d(x, y)\vert y\in\mathcal{X}\setminus S\}, & \mbox{if $x$ $\in S$},
	\end{cases}                        
	\label{sign}
	\end{equation}
where $d$ is a metric on $\mathcal{X}$ and $cl(S)$ denotes the closure of the set $S$. In this paper, we use the metric $d(x,y)=\norm{x-y}$, where $\left\Vert\cdot\right\Vert $ denotes the 2-norm. 

We use $\left[\left[\varphi\right]\right](s_\rho(\cdot;x_0,u), t)$ to denote the robustness degree of the output trajectory $s_\rho(\cdot;x_0,u)$ with respect to the formula $\varphi$ at time $t$. We denote $-\mathcal{I}\triangleq[-i_{2},-i_{1}]$ when $\mathcal{I}=[i_{1},i_{2}]$. The robust semantics of a formula $\varphi$ with respect to $s_\rho(\cdot;x_0,u)$ are defined recursively as follows~\cite{Dokhanchi2014}:
	\begin{align}  
	\begin{split}
	\left[\left[\top\right]\right](s_\rho(\cdot;x_0,u), t):=& +\infty,\\
	\left[\left[\pi\right]\right](s_\rho(\cdot;x_0,u), t):=&\textbf{Dist$_d(s_\rho(\cdot;x_0,u)(t),\mathcal{O}(\pi))$},\\
	\left[\left[\neg\varphi\right]\right](s_\rho(\cdot;x_0,u), t):=&-\left[\left[ \varphi\right]\right](s_\rho(\cdot;x_0,u), t),\\
	\left[\left[\varphi_1\wedge\varphi_2\right]\right](s_\rho(\cdot;x_0,u), t):=&\min\big(\left[\left[ \varphi_1\right]\right](s_\rho(\cdot;x_0,u), t),\\&\left[\left[\varphi_2\right]\right](s_\rho(\cdot;x_0,u), t)\big),\\
	\left[\left[\varphi_1\mathcal{U}_{\mathcal{I}}\varphi_{2}\right]\right](s_\rho(\cdot;x_0,u), t):=&\max_{t'\in (t+\mathcal{I})}\Big(\min\big(\left[\left[\varphi_2\right]\right](s_\rho(\cdot;x_0,u),\\&  t'), \min_{t\le t''<t'}\left[\left[\varphi_1\right]\right]
	(s_\rho(\cdot;x_0,u),t'')\big)\Big).	 
	\end{split}
	\end{align}

                        
\section{Stochastic Control Bisimulation Function}
\subsection{Stochastic Control Bisimulation Function}
We consider the switched stochastic control system with the following dynamics in the mode $q$:
\begin{align}
\begin{split}
& dx=f_q(x,u)dt+g_q(x,u)dw, 
\end{split}
\label{sys}
\end{align}
where the state $x\in\mathcal{X}\in\mathbb{R}^{n}$, the input $u\in\mathcal{V}\in\mathbb{R}^{p}$, $w$ is an $\mathbb{R}^{m}$-valued standard Brownian motion.

Note that the dynamics is essentially the same as that in \cite{Julius2008CDC} when the input signal $u(\cdot)$ is given and bounded, while the existence and uniqueness of the solution of (\ref{sys}) can be guaranteed with the conditions given in \cite{Julius2008CDC}. 

We also consider the switched nominal control system in the mode $q$ as the nominal deterministic version:
\begin{align}
& dx^{\ast}=f_q(x^{\ast},u)dt,  
\label{nom}
\end{align}
 
	\begin{definition}	
		A continuously differentiable function $\psi_q:\mathcal{X}\times\mathcal{X}\times\mathbb{T}\rightarrow \mathbb{R}_{\geqslant 0}$ is a \textbf{time-varying control autobisimulation function} of
		the switched nominal control system (\ref{nom}) in the mode $q$ if for any $x, \tilde{x} \in \mathcal{X}$ ($x\neq\tilde{x}$) and any $t\in\mathbb{T}$ there exists a function $u:\mathbb{R}^n\times\mathbb{T}\rightarrow\mathbb{R}^p$ such that $\psi_q(x,\tilde{x},t)>0$, $\psi_q(x,x,t)=0$ and $\frac{\partial{\psi_q(x,\tilde{x},t)}}{\partial{x}}f_q(x,u(x,t))+\frac{\partial{\psi_q(x,\tilde{x},t)}}{\partial{\tilde{x}}}f_q(\tilde{x},u(\tilde{x},t))+\frac{\partial{\psi_q(x,\tilde{x},t)}}{\partial{t}}\leq0$.
	\end{definition} 
	
In the following, we extend the concept of control autobisimulation function to the stochastic setting.
	
	\begin{definition}	
		A twice differentiable function $\phi_q$: $\mathcal{X}\times\mathcal{X}\rightarrow \mathbb{R}_{\geqslant 0}$ is a \textbf{stochastic control bisimulation function} between (\ref{sys}) and its nominal system (\ref{nom}) if it satisfies \cite{zheACC2018wind}
		\begin{align}
		\begin{split}
		&\phi_q(x,\tilde{x},t)>0, \forall x,\tilde{x}\in \mathcal{X},  x\neq\tilde{x}, \forall t\in \mathbb{T},\\
		&\phi_q(x, x, t)=0, \forall x\in \mathcal{X}, \forall t\in \mathbb{T}, 
		\end{split}
		\end{align}
		and there exist $\alpha_q>0$ and a function
		$u:\mathbb{R}^n\times\mathbb{T}\rightarrow \mathbb{R}^p$ such that 
		\begin{align}
		\begin{split}
		&\frac{\partial\phi_q(x,\tilde{x},t)}{\partial x}f_q(x,u(x,t))+\frac{\partial\phi_q(x,\tilde{x},t)}{\partial \tilde{x}}f_q(\tilde{x},u(\tilde{x},t))\\&+\frac{\partial\phi_q(x,\tilde{x},t)}{\partial t}+\frac{1}{2}g_q^{T}(x,u(x,t))\frac{\partial^{2}\phi_q(x,\tilde{x},t)}{\partial x^{2}}g_q(x,u(x,t))\\&<\alpha_q, 
		\end{split}
		\end{align}
		for any $x, \tilde{x}\in\mathcal{X}$. 
	\end{definition}	

The stochastic control bisimulation function establishes a bound between the trajectories of system (\ref{sys}) and its nominal system (\ref{nom}). 

\subsection{Stochastic Control Bisimulation Function for Switched Linear Dynamics}	
In this subsection, we consider the switched stochastic control system with the following linear dynamics in the mode $q$:
\begin{align}
\begin{split}
& dx=(A_qx+B_qu)dt+\Sigma_q dw,                                                                                      
\end{split}
\label{syslinear}
\end{align}
where $A_q\in\mathbb{R}^{n\times n}$, $B_q\in\mathbb{R}^{n\times p}$, $\Sigma_q\in\mathbb{R}^{n\times m}$.

We can construct a stochastic control bisimulation function of the form
	\begin{center}
		$\phi_q(x, \tilde{x},t)=(x-\tilde{x})^{T}M_q(x-\tilde{x})e^{\mu_q t}$, 
	\end{center}
where $M_q$ is a symmetric positive definite matrix. In order for this function to qualify as a stochastic control bisimulation function, we need to have $M_q\succ 0$, and
	\begin{align}
	\begin{split}
	&\frac{\partial\phi_q(x, \tilde{x}, t)}{\partial x}(A_qx+B_qu)+\frac{\partial\phi_q(x,\tilde{x},t)}{\partial \tilde{x}}(A_q\tilde{x}+B_qu)\\ &+\frac{\partial\phi_q(x,\tilde{x},t)}{\partial t}+
	trace(\frac{1}{2}\Sigma_q^{T}(\frac{\partial^{2}\phi_q(x,\tilde{x},t)}{\partial x^{2}})\Sigma_q)\\ =&(x-\tilde{x})^{T}(2M_qA_q+\mu_q M_q)(x-\tilde{x})+trace(\Sigma_q^{T}M_q\Sigma_q)\\
	<&\alpha.
	\label{func}
	\end{split}
	\end{align}
	for some $\alpha_q>0$. If we pick $\alpha_q=trace(\Sigma_q^{T}M_q\Sigma_q)$, the inequality (\ref{func}) becomes a linear matrix inequality (LMI)
	\begin{align}
		& A_q^{T}M_q+M_qA_q+\mu_q M_q\prec 0. 
		\label{LMI} 
	\end{align}
We denote the system trajectory starting from $x_0$ with the input signal $u(\cdot)$ as $\xi({\bm\cdot};x_0,u)$. It can be seen that (\ref{func}) holds for any input signal $u(\cdot)$, so $u(\cdot)$ is free to be designed.
It can also be verified that $\psi_q(x,\tilde{x},t)=\phi_q(x, \tilde{x},t)=(x-\tilde{x})^{T}M_q(x-\tilde{x})$ is also a time-varying control autobisimulation function of the nominal system  
\begin{align}
\begin{split}
& dx^{\ast}=(A_qx^{\ast}+B_qu)dt. 
\end{split}
\label{nomlinear}
\end{align}

\begin{proposition}                                     
	Given the dynamics of (\ref{nomlinear}), $\psi_q(x, \tilde{x},t) = (x-\tilde{x})^TM_q(x-\tilde{x})e^{\mu_q t}$ is an autobisimulation function if the matrix $M_q$ satisfies the following:
	\begin{align}
	\begin{split}
	&~~~~~~~~M_q\succ 0, ~A_q^TM_q+M_qA_q+\mu_q M_q\preceq0. 
	\end{split}                                   
	\end{align}                  
\end{proposition}

\begin{proof} 
	As $e^{\mu_q t}>0$, if $M_q\succ 0$, then for any $x, \tilde{x}\in\mathcal{X}$ and any $t$, we have $\psi_q(x, \tilde{x},t) = (x-\tilde{x})^TM_q(x-\tilde{x})e^{\mu_q t}>0$. If $A_q^TM_q+M_qA_q+\mu_q M_q\preceq0$, we have for any $x, \tilde{x}\in\mathcal{X}$ and any $t$, 
	\begin{align}\nonumber
	\begin{split}
	&\frac{\partial{\psi_q(x, \tilde{x},t)}}{\partial{x}}f_q(x)+\frac{\partial{\psi_q(x, \tilde{x},t)}}{\partial{\tilde{x}}}f_q(\tilde{x})+\frac{\partial{\psi_q(x, \tilde{x},t)}}{\partial{t}}\\
	&=(x-\tilde{x})^T(A_q^TM_q+M_qA_q+\mu_q M_q)(x-\tilde{x})e^{\mu_q t}\le 0.
	\end{split}                                   
	\end{align}   
	So $\psi_q(x, \tilde{x},t) = (x-\tilde{x})^TM_q(x-\tilde{x})e^{\mu_q t}$ is an autobisimulation function of system (\ref{nomlinear}).    
\end{proof}   

We denote the output trajectory of the nominal system starting from $x_0$ with the input signal $u(\cdot)$ as $s_{\rho^{\ast}}({\bm\cdot};x_0,u)$. 

\begin{proposition}
If $\phi_q$ is a stochastic control bisimulation function between the switched stochastic control system (\ref{syslinear}) and its switched nominal control system (\ref{nomlinear}) in the mode $q$, then for any $T>0$,
\begin{align}
& P\{\sup_{0\leq t\leq T}\phi_q(\xi_q^{\ast}(t;x^0_q,u), \xi_q(t;x^0_q,u))<\gamma\}> 1-\frac{\alpha_q T}{\gamma}.
\label{prob}
\end{align}
\end{proposition}

  \begin{proof} 
  Straightforward from Proposition 2.2 of \cite{Julius2008CDC} and (\ref{func}).
  \end{proof}  	
  
It can be seen from (\ref{prob}) that $\phi_q$ provides a probabilistic upper bound for the distance between the states of the switched stochastic control system and its switched nominal control system in the node $q$ in a finite time horizon. We denote $B_{q}(x,\gamma)\triangleq\{\tilde{x}\vert(x-\tilde{x})^TM_q(x-\tilde{x})\le\gamma\}$.
	
\section{Stochastic Controller Synthesis}
	\label{Sec_Feedforward}
We denote the set of states that satisfy the predicate $p$ as 
$\mathcal{O}(p)\subset\mathcal{X}$. In this paper, we consider a fragment of MTL formulas in the following form:
\begin{align} 
\begin{split}
\varphi =& \Box_{[\tau_1,T_{\textrm{end}}]}p_1 \wedge \Box_{[\tau_2,T_{\textrm{end}}]}p_2 \wedge\dots \wedge\Box_{[\tau_{\eta},T_{\textrm{end}}]}p_{\eta},
\end{split}
\label{MTLform1}
\end{align}	
where $\tau_1<\tau_2<\dots\tau_{\eta}\leq T_{\textrm{end}}$, $T_{\textrm{end}}$ is the end of the simulation time, $\mathcal{O}(p_{\eta})\subset\mathcal{O}(p_{\eta-1})\subset\dots\subset\mathcal{O}(p_1)$, each predicate $p_k$ is in the following form:     
\begin{equation}
p_k\triangleq\left(\bigwedge_{\nu=1}^{n_k}a_{k,\nu}^{T}x+c_{k,\nu}^{T}u<b_{k,\nu}\right),
\label{predicate}
\end{equation}
where ${a}_{k,\nu}\in\mathbb{R}^{n}$ and $b_{k,\nu}\in\mathbb{R}$ denote the                    
parameters that define the predicate, $n_k$ is the number of atomic predicates in the           
$k$-th predicate. We constraint $\lVert{a}_{k,\nu}\rVert_{2}=1$ to reduce redundancy.

The MTL formulas in the above-defined form is actually specifying a series of regions to be entered before certain deadlines and stayed thereafter, with larger regions corresponding to tighter deadlines. The MTL formulas in this form is especially useful in power system frequency regulations as discussed in Section \ref{sec_power}. 
	
	
The $\delta_{k,\nu}$-robust modified formula $\hat{\varphi}_{\delta}$ is defined as follows:
	\begin{align}
\hat{\varphi}_{\delta}\triangleq& \Box_{[\tau_1,T_{\textrm{end}}]}\hat{p}_1 \wedge \Box_{[\tau_2,T_{\textrm{end}}]}\hat{p}_2 \wedge\dots \wedge\Box_{[\tau_{\eta},T_{\textrm{end}}]}\hat{p}_{\eta},
\label{MTLform2}
	\end{align}
where each predicate $\hat{p}_k$ is modified from (\ref{predicate}) as follows:
\begin{equation}
\hat{p}_k\triangleq\left(\bigwedge_{\nu=1}^{n_k}a_{k,\nu}^{T}x+c_{k,\nu}^{T}u<b_{k,\nu}-\Delta_{k,\nu}(t)\right),
\label{predicate2}
\end{equation} 
where
\[
 	\Delta_{k,\nu}(t) =\begin{cases}
 	\delta^0_{k,\nu}e^{-\mu_{q^0} t/2}, ~~~~~~~~~~~~~~~~~~~~~~~~\mbox{if $t<T^0$},\\
 	\delta^i_{k,\nu}e^{-\mu_{q^i} (t-\sum_{j=0}^{i-1}T^j)/2},\\
 	~~~~\mbox{if $\sum_{j=0}^{i-1}T^j\le t<\sum_{j=0}^{i}T^j$, $1\le i\le N_{q}$}.
 	\end{cases} \\
\]                                                                     
	
	\begin{theorem}	                                                        
		If for every $k\in\{1,\dots,\eta\}$ and $\nu\in\{1,\dots,n_k\}$, there exist $z^i_{k,\nu} (i=0,1,\dots,N_q), \epsilon>0$ such that $(z_{k,\nu}^{i})^2a_{k,\nu}a_{k,\nu}^{T}\preceq M_{q^i}$ and $\left[\left[\varphi_{\hat{\delta}}\right]\right](s_{\rho^{\ast}}({\bm\cdot};x^{\ast}_0,u), 0)\ge 0$, where  $\rho^{\ast}=\{(q^{i},\xi^{\ast}_{q^{i}}(t;x^{\ast 0}_{q^i},u),T^{i})\}_{i=0}^{N_{q}}$ is a trajectory of the nominal system, $\varphi_{\hat{\delta}}$ is the $\hat{\delta}_{k,\nu}$-robust modified formula of $\varphi$, $\hat{\delta}^i_{k,\nu}=(\sqrt{r_{q^i}}+\sqrt{\hat{\gamma}})/z^i_{k,\nu}$, $\hat{\gamma}=\frac{(\max\limits_{i}\alpha_{q^i})\cdot T_{\textrm{end}}}{\epsilon}$, 
		\[
		B_{q^{i-1}}(\xi_{q^{i-1}}(\tau;x^{\ast 0}_{q^{i-1}},u),r_{q^{i-1}}e^{-\mu_{q^{i-1}}T^{i-1}/2})\subset B_{q^{i}}(x^{\ast 0}_{q^i},r_{q^{i}}),
		\] 
		then for any $\tilde{x}_0\in B_{q^0}(x^{\ast}_0,r_{q^0})$, the output trajectory $s_{\tilde{\rho}}({\bm\cdot};\tilde{x}_0,u)$ of trajectory $\tilde{\rho}=\{(q^{i},\xi_{q^{i}}(t;\tilde{x}^{0}_{q^i},u),T^{i})\}_{i=0}^{N_{q}}$ satisfies MTL specification $\varphi$ with probability at least $1-\epsilon$, i.e. $P\{\left[\left[\varphi\right]\right](s_{\tilde{\rho}}({\bm\cdot};\tilde{x}_0,u), 0)\ge 0\}>1-\epsilon$.
		\label{th1}
	\end{theorem}	
	\begin{proof} 
		See Appendix.
	\end{proof}  	
	                       
From Theorem \ref{th1}, if we can design the input signal $u(\cdot)$ such that the nominal trajectory $s_{\rho^{\ast}}({\bm\cdot};x_0,u)$ of the nominal system (\ref{nom}) satisfies the $\hat{\delta}_{k,\nu}$-robust modified formula of $\varphi$ (here $\hat{\delta}^i_{k,\nu}\triangleq(\sqrt{r_{q^i}}+\sqrt{\hat{\gamma}})/z^i_{k,\nu}$ for each mode $q^i$, and $\hat{\delta}_{k,\nu}\triangleq[\hat{\delta}^0_{k,\nu}, \dots, \hat{\delta}^{N_q}_{k,\nu}]$), then all the trajectories of the system (\ref{sys}) starting from the initial set $B_{q^0}(x_0,r_{q^0})$ are guaranteed to satisfy the MTL specification $\varphi$ with probability at least $1-\epsilon$. To make the robust modification as tight as possible, for every $k\in\{1,\dots,\eta\}$ and $\nu\in\{1,\dots,n_k\}$, we compute the maximal $z^i_{k,\nu}$ such that $(z_{k,\nu}^{i})^2a_{k,\nu}a_{k,\nu}^{T} \preceq M_{q^{i}}$. We denote the maximal value of $z^i_{k,\nu}$ as $z^{i\ast}_{k,\nu}$, $\hat{\delta}^{i\ast}_{k,\nu}\triangleq(\sqrt{r_{q^{i}}}+\sqrt{\hat{\gamma}})/z^{i\ast}_{k,\nu}$, and the $\hat{\delta}^{\ast}_{k,\nu}$-robust modified formula as $\varphi_{\hat{\delta}^{\ast}}$ (the predicates in $\varphi_{\hat{\delta}^{\ast}}$ are denoted as $\hat{p}^{\ast}_k$). 
	
The optimization problem to find the optimal input signal such that the nominal trajectory satisfies the $\delta^{\ast}_{k,\nu}$-robust modified formula $\varphi_{\hat{\delta}^{\ast}}$ is formulated as follows: 
	\begin{align}
	\begin{split}
	\underset{u(\cdot)}{\argmin} ~ & J(u(\cdot)) \\
	\text{subject to} ~ &\left[\left[\varphi_{\hat{\delta}^{\ast}}\right]\right](s_{\rho^{\ast}}({\bm\cdot};x^{\ast}_0,u), 0)\ge 0.
	\end{split}
	\label{opt}
	\end{align}

The performance measure $J(u(\cdot))$ can be set as the control effort $\norm{u(\cdot)}_1$ (or $\norm{u(\cdot)}_2$). For linear systems, the above optimization problem can be converted to a a mixed-integer linear (or quadratic) programming problem, which can be more efficiently solved using techniques such as McCormick's relaxation \cite{McCormick1976,Gupte2013SolvingMI}. Furthermore, if the MTL formula $\varphi_{\hat{\delta}^{\ast}}$ consists of only conjunctions ($\wedge$) and the always operator ($\Box$), the integers in the optimization problem can be eliminated \cite{sayan2016} and the problem becomes a linear (or quadratic) programming problem. 


\section{Case Study on Power Systems}
\label{sec_power}
In this section, we implement the proposed controller synthesis approach in two scenarios in power systems.
\subsection{Scenario I}
In this subsection, we implement the controller synthesis method for regulating the grid frequency of a four-bus system with a 600 MW thermal plant $\textrm{G}_{1}$ made up of four identical units, a wind farm $\textrm{G}_{2}$ consisting of 200 identical 1.5 MW Type-C wind turbine generators (WTG) and an energy storage system (ESS), as shown in Fig. \ref{dfig}. The configuration parameters of each Type-C WTG can be found in Appendix B of \cite{PulgarPhd} (we set $C_{\mathrm{opt}}=16.1985\times10^{-9}~[s^3/rad^3]$).

\begin{figure}[ht]
	\centering
	\includegraphics[scale=0.3]{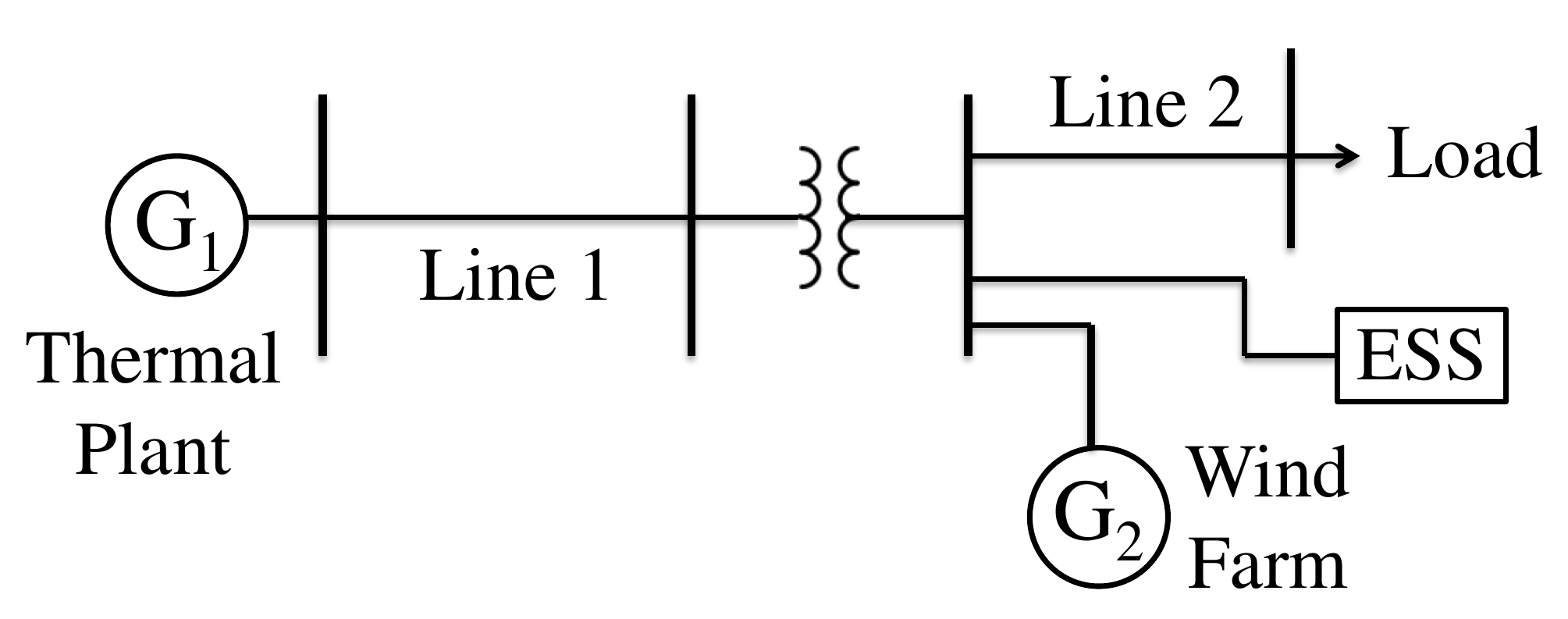}
	\caption{The four-bus system \cite{ZhangPulgar2017} with a thermal plant, a wind farm and an energy storage system (ESS).} 
	\label{dfig}
\end{figure}

By linearizing the system of differential-algebraic equations at the equilibrium point, we have

	\begin{align}
	\begin{split}                                            
	& d\begin{bmatrix}
	\Delta x \\
       0\\		
	\end{bmatrix}=
	\begin{bmatrix}
	A_s& B_s \\
	C_s& D_s \\		
	\end{bmatrix}
	\begin{bmatrix}
	\Delta x \\
    \Delta y\\		
	\end{bmatrix}dt+
	\begin{bmatrix}
	 M_s \\
	 N_s\\		
	\end{bmatrix}u^wdt+\begin{bmatrix}
	\Sigma_{s1} \\
	\Sigma_{s2}\\		
	\end{bmatrix}dw,\\
	&\Delta P_{\rm{gen}}=
	\begin{bmatrix}
	E_s ~~
	F_s\\		
	\end{bmatrix}
	\begin{bmatrix}
	\Delta x \\
	\Delta y \\		
	\end{bmatrix},
	\end{split}
	\end{align}
where $w$ is standard Brownian motion representing the stochasticity of the wind, $u^w$ is a control input through which the wind turbine generator can adjust its power output, $\Delta x=[\Delta E'_{qD}, \Delta E'_{dD}, \Delta\omega_r, \Delta x_1, \Delta x_2, \Delta x_3, \Delta x_4]^T$, $\Delta y=[\Delta P_{\rm{gen}}, \Delta Q_{\rm{gen}}, \Delta V_{dr}, \Delta V_{qr}, \Delta I_{dr}, \Delta I_{qr}, \Delta I_{ds}, \Delta I_{qs},$ $\Delta V_D, \Delta \theta_D]^T$, $\Delta E_{dD}'$, $\Delta E_{qD}'$ and $\Delta \omega_{r}$ are the $d$, $q$ axis voltage variation and rotor speed variation of the WTG, respectively, $\Delta x_{1}$ to $\Delta x_{4}$ are variations of proportional-integral (PI) regulator induced states, $\triangle P_{\mathrm{gen}}$ and $\triangle Q_{\mathrm{gen}}$ are the active and reactive power variation from each WTG, $\triangle V_{dr}, \triangle V_{qr}, \triangle I_{dr}, \triangle I_{qr}$ are the rotor $d$, $q$ axis voltage and current variation, respectively, $\triangle I_{ds}, \triangle I_{qs}$ are the stator $d$, $q$ axis current variation, respectively, and $A_s$, $B_s$, $C_s$, $D_s$, $E_s$, and $F_s$ are matrices from the linearization at the equilibrium point.

Through the Kron reduction, we have
\begin{align}
\begin{split}
& d\Delta x=A_{\rm{kr}}\Delta xdt+B_{\rm{kr}}u^wdt+\Sigma_{\rm{kr}}dw,                   \\
&\Delta P_{\rm{gen}}=C_{\rm{kr}}\Delta x+D_{\rm{kr}}u^w+E_{\rm{kr}}dw/dt,
\end{split}
\label{Kron}
\end{align}                 
where 
\begin{center}
~~~~$A_{\rm{kr}}=A_s-B_sD_s^{-1}C_s$, ~~~
$B_{\rm{kr}}=M_s-B_sD_s^{-1}N_s$, ~~~
~~~~$C_{\rm{kr}}=E_s-F_sD_s^{-1}C_s$, ~~~
$D_{\rm{kr}}=-F_sD_s^{-1}N_s$,
~~~~$\Sigma_{\rm{kr}}=\Sigma_{s1}-B_sD_s^{-1}\Sigma_{s2}$, ~~
$E_{\rm{kr}}=-F_sD_s^{-1}\Sigma_{s2}$.
\end{center}

We consider a disturbance of generation loss of 150 MW (loss of one unit), denoted as $\Delta P_{d}=0.15$, that occurs at time 0. From 0 second to 5 seconds after the disturbance, the system frequency response model of the the four-bus system is as follows (we choose base MVA as 1000MVA):
\begin{align}
\begin{cases}
&\triangle\dot{\omega}=\frac{\omega_{s}}{2H}(\triangle P_{m}+u^s+\triangle P_s-\triangle P_d+200\Delta P_{\rm{gen}}/1000\\&~~~~~~~-\frac{D}{\omega_{s}}\triangle\omega),                    \\
&\triangle\dot{P}_s=0;\\
&\triangle\dot{P}_{m}=\frac{1}{\tau_{\mathrm{c}\mathrm{h}}}(\triangle P_{v}-\triangle P_{m}),\\
&\triangle\dot{P}_{v}=\frac{1}{\tau_{g}}(-\triangle P_{v}-\frac{1}{2\pi R}\triangle\omega),
\end{cases}
\label{SG1}
\end{align}                
where $u^s$ is a control input representing the power injection from the energy storage system (ESS), $\triangle\omega$ is the grid frequency deviation, $\triangle P_{m}$ is the governor mechanical power variation, $\triangle P_{v}$ is the governor valve position variation, $\triangle P_s$ is the variation of the generator power re-dispatch, and $\triangle P_{d}$ denotes a large disturbance. $\triangle P_{\mathrm{gen}}$ times 200 as there are 200 WTGs, and it is divided by 1000 as the base MVA for each WTG and the power system are 1 MVA and 1000 MVA, respectively. We set $\omega_{s}=2\pi\times60$rad/s, $D$=1, $H$=4s, $\tau_{\mathrm{ch}}$=0.3s, $\tau_{g}$=0.1s, $R$=0.05.  

From 5 seconds to 8.75 seconds after the disturbance, the generator power re-dispatch ($\triangle P_s$) starts with a ramping rate of 0.04 with the following system frequency response model.

\begin{align}
\begin{cases}
&\triangle\dot{\omega}=\frac{\omega_{s}}{2H}(\triangle P_{m}+u^s+\triangle P_s-\triangle P_d+200\Delta P_{\rm{gen}}/1000\\&~~~~~~~-\frac{D}{\omega_{s}}\triangle\omega),                    \\
&\triangle\dot{P}_s=0.04;\\
&\triangle\dot{P}_{m}=\frac{1}{\tau_{\mathrm{c}\mathrm{h}}}(\triangle P_{v}-\triangle P_{m}),\\
&\triangle\dot{P}_{v}=\frac{1}{\tau_{g}}(-\triangle P_{v}-\frac{1}{2\pi R}\triangle\omega),
\end{cases}
\label{SG2}
\end{align}     

At 8.75 seconds, the generation and load are balanced again. So from 8.75 seconds to 10 seconds after the disturbance, the system frequency response model is the same as that in (\ref{SG1}).

With (\ref{Kron}), (\ref{SG1}) and (\ref{SG2}), we have the following switched stochastic control system with two modes corresponding to (\ref{SG1}) and (\ref{SG2}) respectively:        
\begin{align}
d\hat{x}=(\hat{A}_q\hat{x}+\hat{B}_qu)dt+\hat{\Sigma}_qdw,
\label{whole}
\end{align}  
where $\hat{x}=[\triangle E'_{qD}, \triangle E'_{dD}, \triangle\omega_r, \triangle x_1, \triangle x_2, \triangle x_3, \triangle x_4,\triangle\omega,$ $\triangle P_s, \triangle P_{m},\triangle P_{v}]^T$, the input $u=[u^w, u^s]^T$. As the matrix $\hat{A}_q$ is computed as Hurwitz for both modes, the system in each mode is stable.             

We use the following MTL specification for frequency regulation after the disturbance:	  
\begin{align}
\begin{split}
\varphi =& \Box_{[0,T_{\textrm{end}}]}p_1 \wedge \Box_{[2,T_{\textrm{end}}]}p_2,\\
p_1 =& (-0.5{\rm Hz}\leq \Delta f\leq 0.5{\rm Hz})\wedge(-10{\rm Hz}\leq \Delta f_{r}\leq 10{\rm Hz}),\\
p_2=& (-0.4{\rm Hz}\leq \Delta f\leq 0.4{\rm Hz}),
\end{split}
\label{spec}
\end{align}	
where $\Delta f=\frac{\triangle\omega}{2\pi}$, $\Delta f_{r}=\frac{\triangle\omega_{r}}{2\pi}$. The specification means ``After a disturbance, the grid frequency deviation should never exceed $\pm$0.5Hz, the WTG rotor speed deviation should never exceed $\pm$10Hz, after 2 seconds the grid frequency deviation should always be within $\pm$0.4Hz''.  
	

	\begin{table}[]
		\label{{parameter}}  
		\centering
		\caption{System parameters}
		\label{my-label}
		\begin{tabular}{lll}
			\toprule[2pt] 
			VA base $P_\textrm{b}$                 &       & 1000MVA \\ \hline
			System frequency $f_\textrm{s}$            &       & 60Hz   \\ \hline                        
			Active power flow to load $\textrm{L}_{i}$         & $i$=1       & 0.4 (pu)   \\ \hline
										& $i$=2       & 0.1 (pu)   \\ \hline
										& $i$=3       & 0.05 (pu)    \\ \hline
										& $i$=4       & 0.05 (pu)    \\ \hline
			Transformer impedance                    & $\textrm{G}_1$          & 1.8868 (pu)     \\ \hline                                   & $\textrm{G}_2$          & 0.618 (pu)    \\ \bottomrule[2pt] 
		\end{tabular}
	\end{table}     
	
	\begin{table}[]
		\label{{parameter}}  
		\centering
		\caption{Line data (1000 MVA base).}
		\label{my-label}
		\begin{tabular}{lll}
			\toprule[2pt]    
			Line number    & Line impedance (pu)   & Line charging (pu) \\ \hline
			2-8(2-9)          & j0.01 & 0.0006625\\ \hline
			7-8(7-9)          & j0.04 & 0.0023 \\ \hline
			4-8(4-9)          & j0.03 & 0.0031 \\ \hline
			4-5               & j0.03 & 0.0034  \\ \hline
			5-6               & j0.03 & 0.0094 \\ \hline
			6-7               & j0.02 & 0.0258 \\ \bottomrule[2pt]      
		\end{tabular}
	\end{table}

We set $k_w=1$, $T_{\textrm{end}}=5$ (s), $\epsilon=\alpha T_{\textrm{end}}/\hat{\gamma}=5\%$, so $\alpha=0.05\hat{\gamma}/T_{\textrm{end}}=0.01\hat{\gamma}$. As $\alpha=trace(\hat{\Sigma}^{T}M\hat{\Sigma})=k_w^2M(3,3)$, we have $\hat{\gamma}=100k_w^2M(3,3)=100M(3,3)$. We assume that the initial state variations can be covered by $B_{q^0}(\hat{x}^{\ast}_0,r)$, where $r=4\hat{\gamma}$ ($4=2^2$ is chosen as the initial state variations due to the time needed for running the algorithm to generate the controller, which is about twice the simulation time), $\hat{x}^{\ast}_0$ is zero in every dimension. It can be seen from (\ref{spec}) that the allowable variation range of the grid frequency variation $\triangle\omega$ is much smaller than that of the wind turbine rotor speed variation $\triangle\omega_{r}$. Therefore, in order to decrease the conservativeness of the probabilistic bound as much as possible, we further optimize both $z_{k,i}$ and the matrix $M_{q^i}$ such that the outer bounds of the stochastic robust neighbourhoods in the dimension of the grid frequency variation $\hat{\delta}^{i\ast}_{1,1}$ ($\hat{\delta}^{i\ast}_{1,1}= \hat{\delta}^{i\ast}_{1,2}=\hat{\delta}^{i\ast}_{2,1}=\hat{\delta}^{i\ast}_{2,2}$) are much smaller than the outer bounds in the dimension of the wind turbine rotor speed variation $\hat{\delta}^{i\ast}_{1,3}$ ($\hat{\delta}^{i\ast}_{1,3}= \hat{\delta}^{i\ast}_{1,4}$). As $\hat{\delta}_{k,i}=(\sqrt{r_{q^i}}+\sqrt{\hat{\gamma}})/z_{k,i}$ and $\hat{\gamma}=100M_{q^0}(3,3)$, minimizing $\hat{\delta}_{1,1}$ can be achieved by minimizing $M_{q^i}(3,3)~(i=1,2,3)$ and maximizing $z^i_{1,1}$. Therefore, to obtain both $M^{\ast}$ and $z^{\ast}_{1,1}$, we solve the following semidefinite programming (SDP) problem as follows.

\begin{align}
\begin{split}
&\textrm{min}. -(z^i_{1,1})^2\\	  
\textrm{s.t.} ~&  M_{q^i}\succ 0,
\hat{A}_{q^i}^{T}M_{q^i}+M_{q^i}\hat{A}_{q^i}+\mu_{q^i} M_{q^i}\preceq 0,\\
& e_3^{T}M_{q^i}e_3\leq \zeta, M_{q^i}-(z^i_{1,1})^2a_{1,1}a_{1,1}^{T}\succeq 0.
\end{split}
\label{opt1}
\end{align}
where $e_3=[0,0,1,0,0,0,0,0,0,0]^T$, $\mu_{q^i}=0.1$, $\zeta$ is tuned manually to be as small as possible while the optimization problem is feasible.

With the $M_{q^i}^{\ast}$ obtained from (\ref{opt1}), we compute the tightest outer bound in the dimension of $\triangle\omega_{r}$ as follows:

\begin{align}
\begin{split}
&\textrm{min}. -(z^i_{1,3})^2\\	  
 \textrm{s.t.} ~  
& M_{q^i}^{\ast}-(z^i_{1,3})^2a_{1,3}a_{1,3}^{T}\succeq 0.
\end{split}
\label{opt2}
\end{align}

From (\ref{opt1}) and (\ref{opt2}), we obtain the $\hat{\delta}^{\ast}_{k,i}$-robust modified formula as follows.
\[
\begin{split}
\varphi_{\hat{\delta}^{\ast}} =& \Box_{[0,T_{\textrm{end}}]}\hat{p}^{\ast}_1 \wedge \Box_{[2,T_{\textrm{end}}]}\hat{p}^{\ast}_2,\\
\hat{p}^{\ast}_1 =& (-0.5{\rm Hz}+0.217 e^{-0.01t}{\rm Hz}\leq \Delta f\\&\leq 0.5{\rm Hz}-0.217 e^{-0.01t}{\rm Hz})\wedge\\
&(-10{\rm Hz}+6.08 e^{-0.01t}{\rm Hz}\leq \Delta f_{r}\\&\leq 10{\rm Hz}-6.08 e^{-0.01t}{\rm Hz}),\\
\hat{p}^{\ast}_2=& (-0.4{\rm Hz}+0.217 e^{-0.01t}{\rm Hz}\leq \Delta f\\&\leq 0.4{\rm Hz}-0.217 e^{-0.01t}{\rm Hz}).
\end{split}
\]

\begin{figure}[th]
	\centering
	\includegraphics[width=6cm]{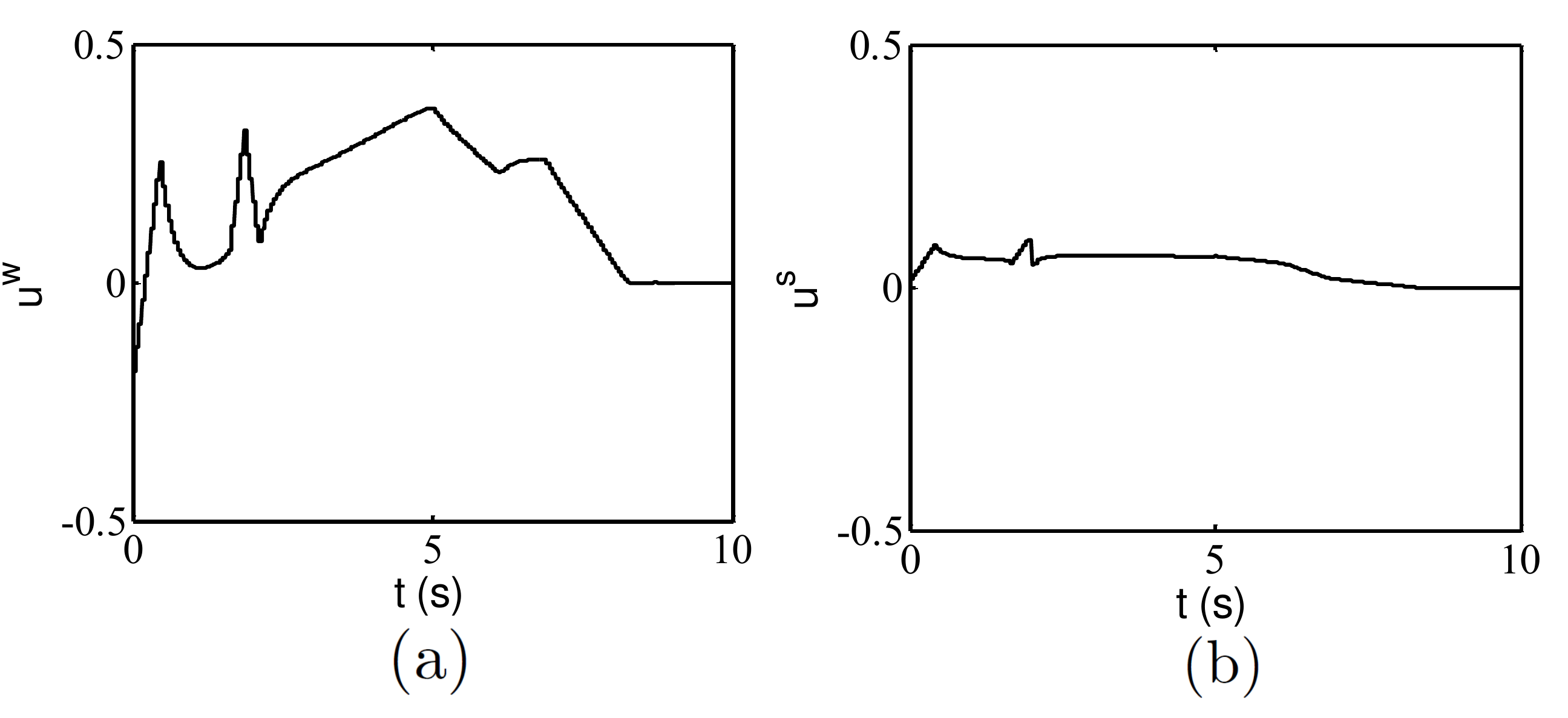}\caption{The synthesized optimal input signals in Scenario I.}
	\label{wind_u}
\end{figure}          

\begin{figure}[th]
	\centering
	\includegraphics[width=6cm]{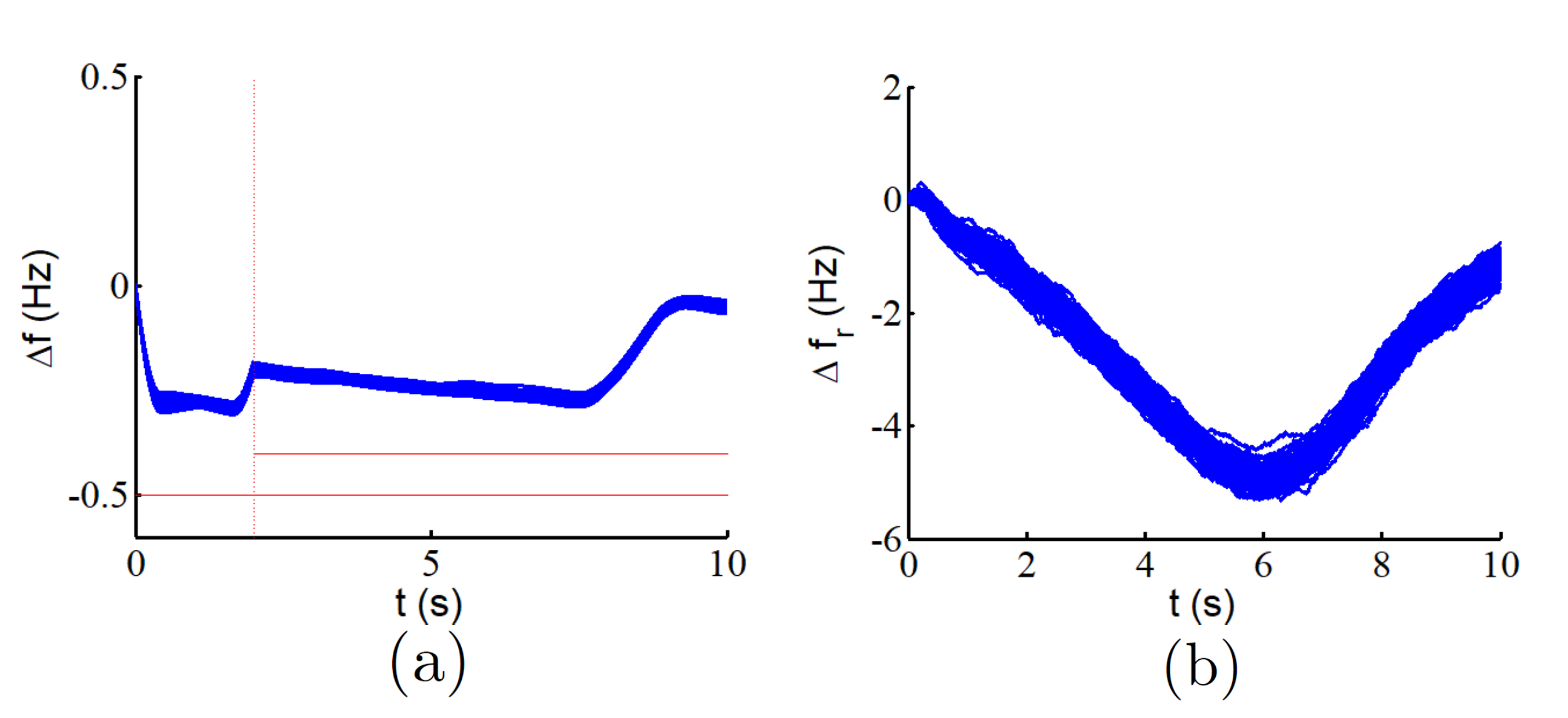}\caption{100 trajectories (realizations) of $\Delta f$ and $\Delta f_r$ with the synthesized control inputs (blue) in Scenario I. The red lines indicate the thresholds in $\varphi$.}
	\label{wind_f}
\end{figure}

We perform the controller synthesis with respect to $\varphi_{\hat{\delta}^{\ast}}$. We set $J(u(\cdot))=\norm{u^w(\cdot)}_2+\lambda\norm{u^s(\cdot)}_2$, where $\lambda=100$ (we encourage power input from the WTGs). Fig. \ref{wind_u} shows the~computed optimal input signals. Fig. \ref{wind_f} shows that all 100 trajectories (realizations) starting from $B_{q^0}(x^{\ast}_0,r)$ with the optimal input signals satisfy the specification $\varphi$.

\subsection{Scenario II}
In this section, we apply the controller synthesis method on a nine-bus system as shown in Fig. \ref{gridnew}. The thermal plant $\textrm{G}_{1}$ and the wind farm $\textrm{G}_{2}$ are the same as those in Scenario I, with two energy storage systems (ESS) placed near them respectively. Four constant power loads are denoted as $\textrm{L}_1$, $\textrm{L}_2$, $\textrm{L}_3$ and $\textrm{L}_4$. The line data can be found in Tab. I and Tab. II \cite{zhe_control}. We consider a disturbance of generation loss of 150 MW (loss of one unit, $\Delta P_{d}=0.15$) that occurs at time 0. The switched stochastic system can be written in a similar form as in (\ref{whole}), with the modes transitioning at 5 seconds and 8.75 seconds, respectively.

\begin{figure}
	\centering
	\includegraphics[scale=0.2]{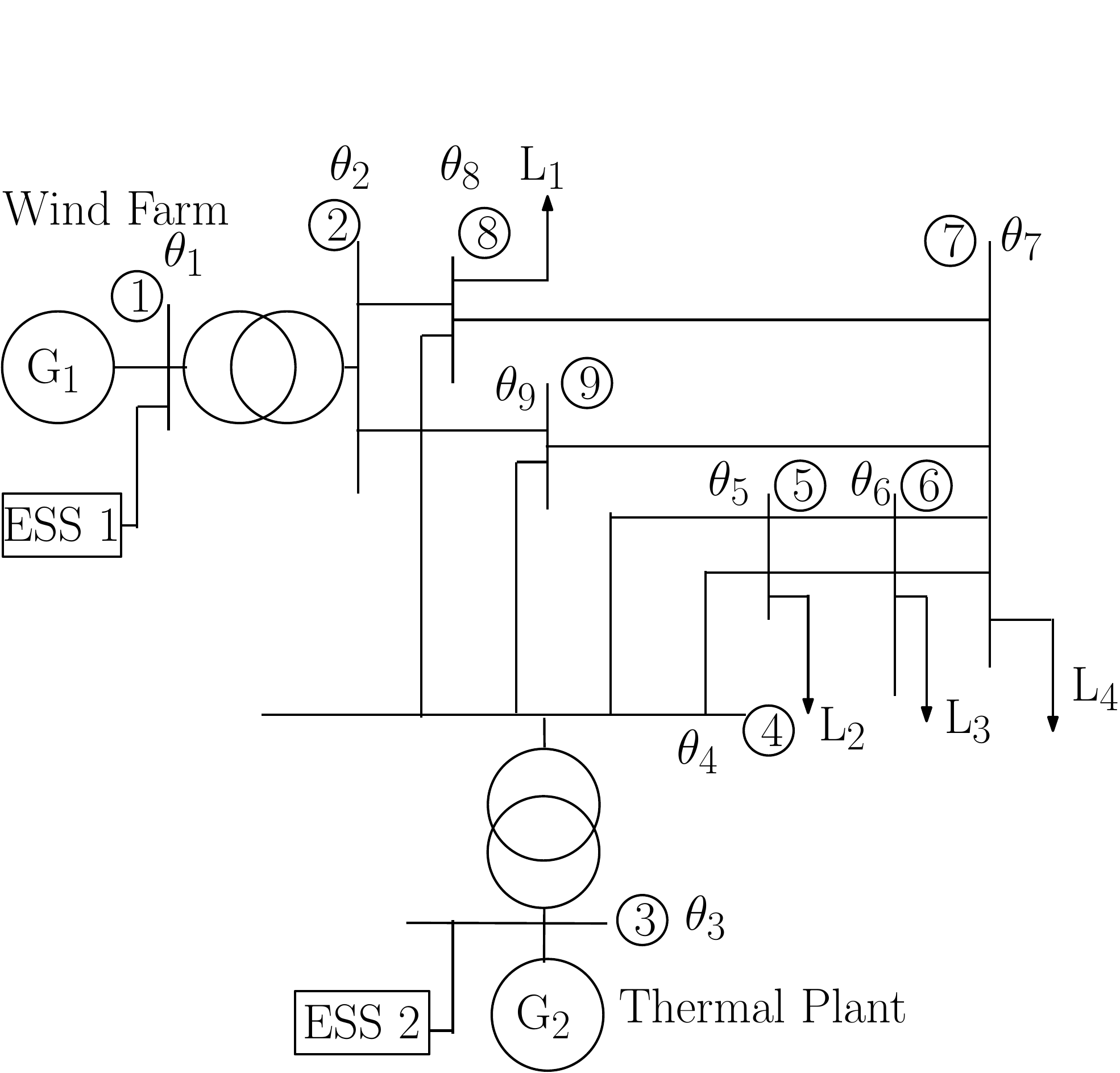}
	\caption{The nine-bus system with a thermal plant, a wind farm and two energy storage systems (ESS).} 
	\label{gridnew}
\end{figure}

We use the following MTL specification for frequency regulation after the disturbance (note that here in Scenario II, we use $\bar{\varphi}$ to show difference with $\varphi$ in Scenario I):	  
\begin{align}
\begin{split}
\bar{\varphi} =& \Box_{[0,T_{\textrm{end}}]}\bar{p}_1 \wedge \Box_{[2,T_{\textrm{end}}]}\bar{p}_2\wedge\Box_{[2,T_{\textrm{end}}]}\bar{p}_3,\\
\bar{p}_1 =& (-0.5{\rm Hz}\leq \Delta f\leq 0.5{\rm Hz})\wedge(-10{\rm Hz}\leq \Delta f_{r}\leq 10{\rm Hz}),\\
\bar{p}_2=& (-0.4{\rm Hz}\leq \Delta f\leq 0.4{\rm Hz}),\\
\bar{p}_3=&\bigwedge_{ij\in\mathcal{E}}(-0.25\leq P_{ij}\leq 0.25).
\end{split}
\label{spec2}
\end{align}	

The first two subformulas in (\ref{spec2}) are the same as in (\ref{spec}) used in Scenario I, while the third subformula $\Box_{[2,T_{\textrm{end}}]}\bar{p}_3$ specifies the real power constraints in each line. We obtain the following $\hat{\delta}^{\ast}_{k,i}$-robust modified formula:
\[
\begin{split}
\bar{\varphi}_{\hat{\delta}^{\ast}} =& \Box_{[0,T_{\textrm{end}}]}\hat{\bar{p}}^{\ast}_1 \wedge \Box_{[2,T_{\textrm{end}}]}\hat{\bar{p}}^{\ast}_2 \wedge\Box_{[2,T_{\textrm{end}}]}\hat{\bar{p}}^{\ast}_3,\\
\hat{\bar{p}}^{\ast}_1 =& (-0.5{\rm Hz}+0.217 e^{-0.01t}{\rm Hz}\leq \Delta f\\&\leq 0.5{\rm Hz}-0.217 e^{-0.01t}{\rm Hz})\wedge\\
&(-10{\rm Hz}+6.08 e^{-0.01t}{\rm Hz}\leq \Delta f_{r}\\&\leq 10{\rm Hz}-6.08 e^{-0.01t}{\rm Hz}),\\
\hat{\bar{p}}^{\ast}_2=& (-0.4{\rm Hz}+0.217 e^{-0.01t}{\rm Hz}\leq \Delta f\\&\leq 0.4{\rm Hz}-0.217 e^{-0.01t}{\rm Hz}),\\
\hat{\bar{p}}^{\ast}_3=&\bigwedge_{ij\in\mathcal{E}}(-0.25+0.0258 e^{-0.01t}\leq P_{ij}\\&\leq 0.25-0.0258 e^{-0.01t}).
\end{split}
\]
where $\mathcal{E}\subset\mathcal{N}\times\mathcal{N}$ is the set of transmission lines ($\mathcal{N}$ is the set of buses).

As there are 9 different lines corresponding to 9 different inequalities in the MTL specification, solving the optimization problem with all the inequality constraints could be computationally expensive. To reduce computation, we first set an initial MTL specification and iteratively add the line power flow inequality constraints that are violated with the previous optimization. The initial MTL specification $\bar{\varphi}^0_{\hat{\delta}^{\ast}}$ as follows (by reducing the line power flow constraints in $\bar{\varphi}_{\hat{\delta}^{\ast}}$):
\[
\begin{split}                                           
\bar{\varphi}^0_{\hat{\delta}^{\ast}} =& \Box_{[0,T_{\textrm{end}}]}\hat{\bar{p}}^{\ast}_1 \wedge \Box_{[2,T_{\textrm{end}}]}\hat{\bar{p}}^{\ast}_2.
\end{split}
\]

We perform the controller synthesis with respect to $\bar{\varphi}^0_{\hat{\delta}^{\ast}}$. We set $J(u(\cdot))=\norm{u^w(\cdot)}_2+\lambda\norm{u^s(\cdot)}_2$, where $\lambda=100$ (larger $\lambda$ encourages power input from the wind turbine generator). After the first iteration, the line 2-8 is overloaded and thus does not satisfy the line flow constraint in $\bar{\varphi}_{\hat{\delta}^{\ast}}$ (as shown in Fig. \ref{wind_line}). Thus we add line 2-8 power specification and obtain the following MTL specification $\bar{\varphi}^1_{\hat{\delta}^{\ast}}$:
\[
\begin{split}
\bar{\varphi}^1_{\hat{\delta}^{\ast}} =& \Box_{[0,T_{\textrm{end}}]}\hat{\bar{p}}^{\ast}_1 \wedge \Box_{[2,T_{\textrm{end}}]}\hat{\bar{p}}^{\ast}_2\wedge\Box_{[2,T_{\textrm{end}}]}\hat{\bar{p}}^{\ast1}_3,\\
\hat{\bar{p}}^{\ast1}_3=& (-0.25+0.0258 e^{-0.01t}\leq P_{28}\leq 0.25-0.0258 e^{-0.01t}).
\end{split}
\]
In the second iteration, the computed control inputs not only lead to satisfaction of $\bar{\varphi}^1_{\hat{\delta}^{\ast}}$, but also the satisfaction of $\bar{\varphi}_{\hat{\delta}^{\ast}}$. Thus the iteration stops and the optimal input signals are obtained (as shown in Fig. \ref{wind_u2}). Using the same $r$ as that in Scenario I, Fig. \ref{wind_f2} shows that all 100 trajectories (realizations) of $\Delta f$ and $\Delta f_r$ starting from $B_{q^0}(x^{\ast}_0,r)$ with the synthesized optimal input signals satisfy the MTL specification $\bar{\varphi}$.

\begin{figure}[th]
	\centering
	\includegraphics[width=8.5cm]{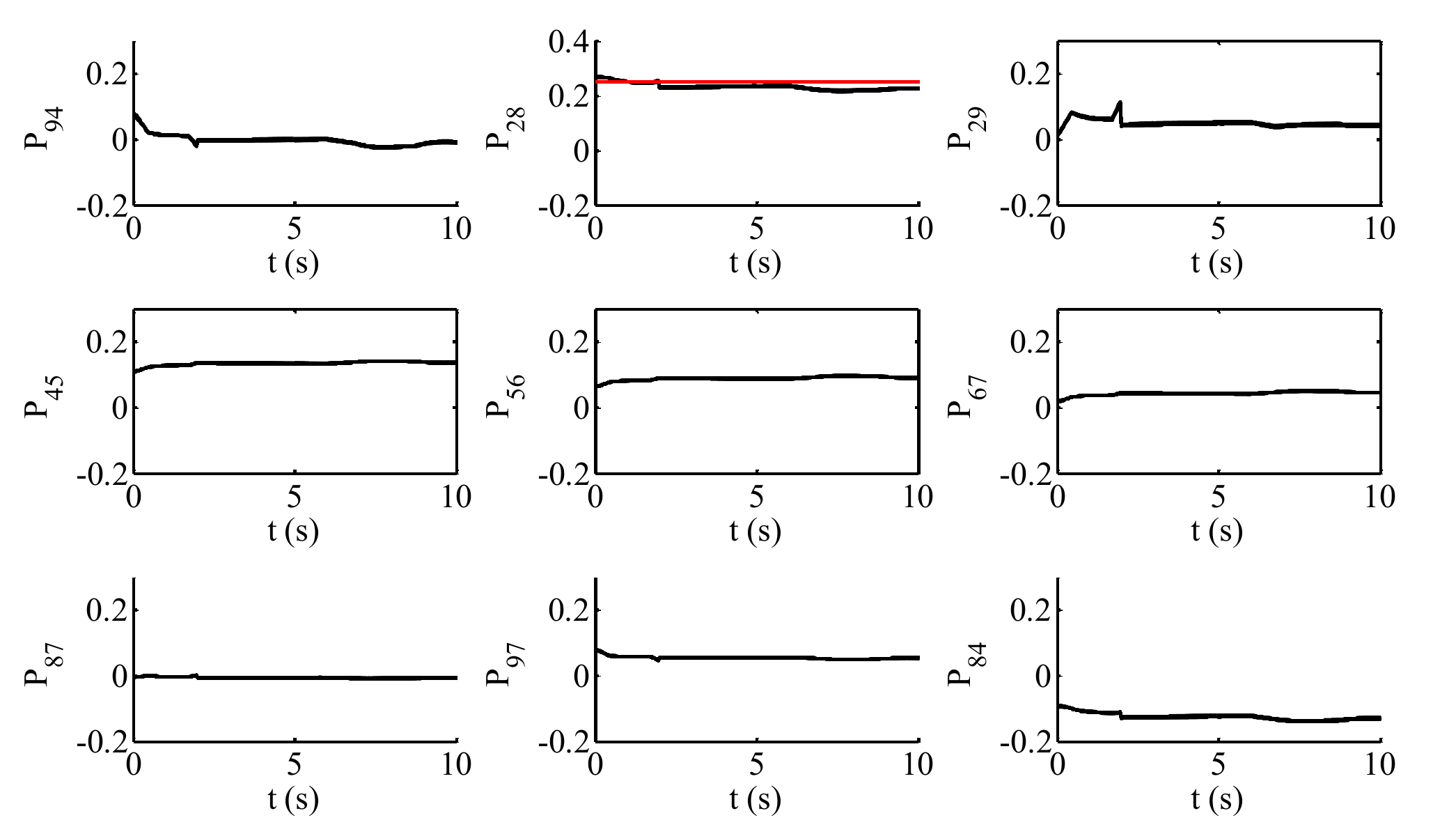}\caption{100 trajectories (realizations) of real power of 9 different lines with the synthesized control inputs (blue) after the first iteration in Scenario II.}
	\label{wind_line}
\end{figure}

\begin{figure}[th]                                          
	\centering
	\includegraphics[width=9cm]{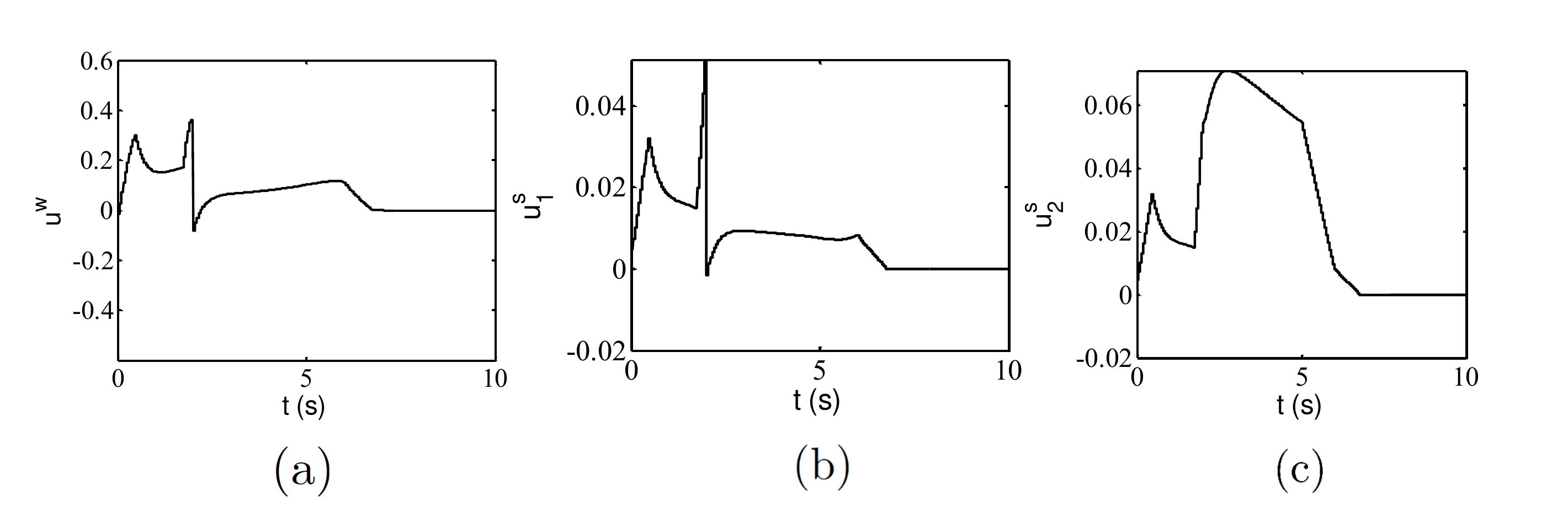}\caption{The synthesized optimal input signals in Scenario II.}
	\label{wind_u2}
\end{figure}

\begin{figure}[th]
	\centering
	\includegraphics[width=6cm]{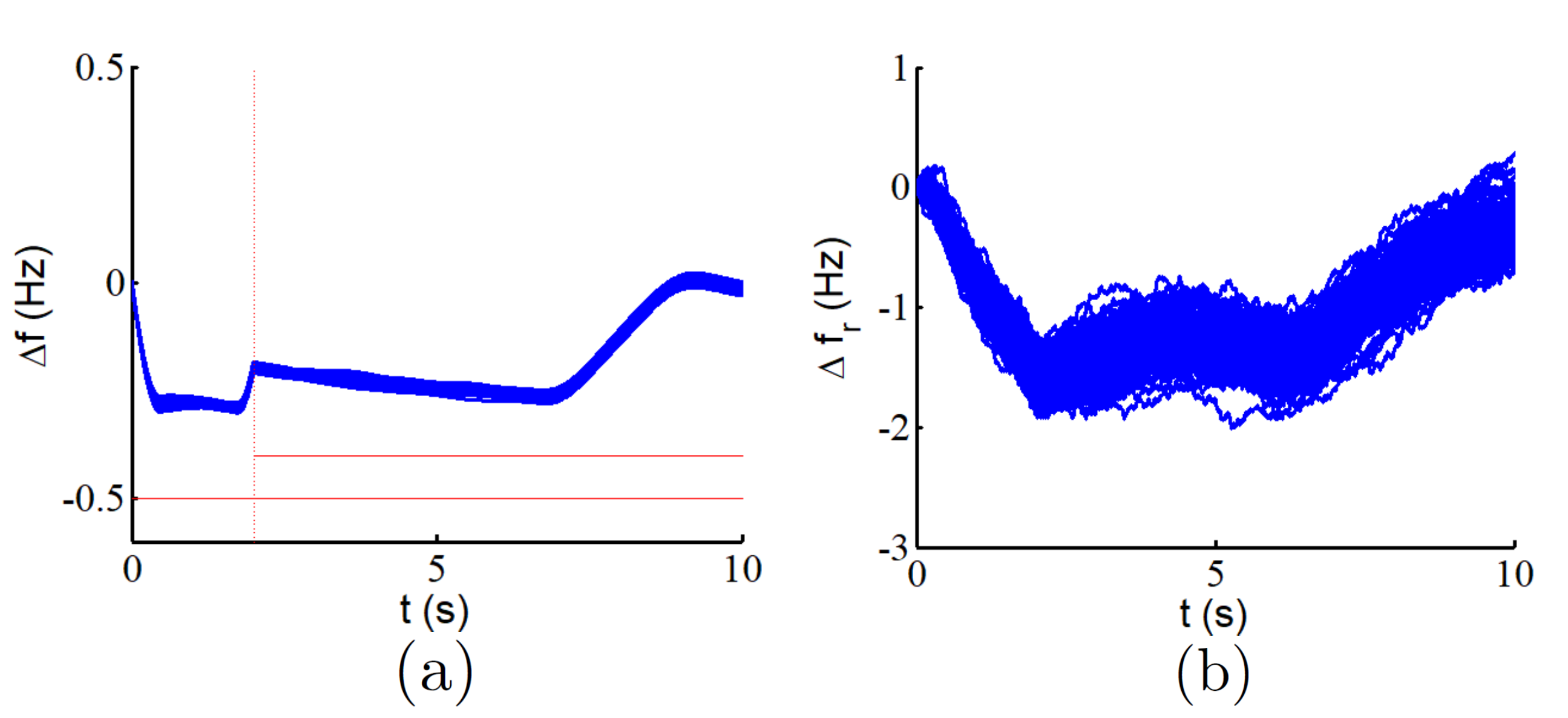}\caption{100 trajectories (realizations) of $\Delta f$ and $\Delta f_r$ with the synthesized control inputs (blue) in Scenario II. The red lines indicate the thresholds in $\bar{\varphi}$.}
	\label{wind_f2}
\end{figure}

\section{CONCLUSIONS}
\label{conclusion}
We presented a provably correct controller synthesis approach for switched stochastic systems with metric temporal logic specifications. We implemented the approach on power systems, while the same approach can be applied in switched control systems in other applications such as robotic systems, communication systems, and biological systems.


\section*{APPENDIX}
	\textbf{Proof of Theorem \ref{th1}}:\\
For the output trajectory $s_{\rho^{\ast}}({\bm\cdot};x^{\ast}_0,u)$ of trajectory $\rho^{\ast}=\{(q^{i},\xi^{\ast}_{q^{i}}(t;x^{\ast 0}_{q^i},u),T^{i})\}_{i=0}^{N_{q}}$ (where $x^{\ast 0}_{q^0}=x^{\ast}_0$) of the switched nominal control system, if $B_{q^{i-1}}(\xi_{q^{i-1}}(T^{i-1},x^{\ast 0}_{q^{i-1}}),r_{q^{i-1}}e^{-\mu_{q^{i-1}} T^{i-1}/2})\subset B_{q^{i}}(x^{\ast 0}_{q^i},$ $r_{q^{i}}) (i=1,2,\dots,N_q)$, then for any $\tilde{x}_0\in B_{q^0}(x^{\ast}_0,r_{q^0})$ and the output trajectory $s_{\tilde{\rho}^{\ast}}({\bm\cdot};\tilde{x}_0,u)$ of trajectory $\tilde{\rho}^{\ast}=\{(q^{i},\xi^{\ast}_{q^{i}}(t;\tilde{x}^{\ast 0}_{q^i},u),T^{i})\}_{i=0}^{N_{q}}$ (where $\tilde{x}^{\ast 0}_{q^0}=\tilde{x}_0$) of the switched nominal control system, we have $\tilde{x}^{\ast0}_{q^i}\in B_{q^i}(x^{\ast 0}_{q^i},r_{q^i})$.

For every $k\in\{1,\dots,\eta\}$, $\nu\in\{1,\dots,n_k\}$ and any $\tilde{x}^{\ast 0}_{q^i}\in B_{q^i}(x^{\ast 0}_{q^i},r_{q^i})$, we have
	\begin{align}
	\begin{split}
	& (\xi_{q^i}^{\ast}(t;\tilde{x}^{\ast 0}_{q^i},u)-\xi_{q^i}^{\ast}(t;x^{\ast 0}_{q^i},u))^{T}a_{k,\nu}a_{k,\nu}^{T}(z_{k,\nu}^{i})^2(\xi_{q^i}^{\ast}(t;\tilde{x}^{\ast 0}_{q^i},u)\\
	&-\xi_{q^i}^{\ast}(t; x^{\ast 0}_{q^i},u))\le (\xi_{q^i}^{\ast}(t;\tilde{x}^{\ast 0}_{q^i},u)-\xi_{q^i}^{\ast}(t;x^{\ast 0}_{q^i},u))^{T}M_{q^i}\\
	&(\xi_{q^i}^{\ast}(t;\tilde{x}^{\ast 0}_{q^i},u)-\xi_{q^i}^{\ast}(t;x^{\ast 0}_{q^i},u))\\&=\psi_{q^i}(\xi_{q^i}^{\ast}(t;\tilde{x}^{\ast 0}_{q^i},u), \xi_{q^i}^{\ast}(t;x^{\ast 0}_{q^i},u))e^{-\mu_{q^i} t}\le r_{q^i}e^{-\mu_{q^i} t}.
	\end{split}  
	\end{align}		
	Therefore, we have $\norm{a_{k,\nu}^{T}(\xi_{q^i}^{\ast}(t;\tilde{x}^{\ast 0}_{q^i},u)-\xi_{q^i}^{\ast}(t;x^{\ast 0}_{q^i},u))}\le\sqrt{r_{q^i}}e^{-\mu_{q^i} t/2}/z^i_{k,\nu}$, thus 
	\begin{align}
	  \begin{split}
	& -\sqrt{r_{q^i}}e^{-\mu_{q^i}t/2}/z^i_{k,\nu}\le a_{k,\nu}^{T}(\xi_{q^i}^{\ast}(t;\tilde{x}^{\ast 0}_{q^i},u)-\xi_{q^i}^{\ast}(t;x^{\ast 0}_{q^i},u))\\&\le\sqrt{r_{q^i}}e^{-\mu_{q^i} t/2}/z^i_{k,\nu}. 
       \end{split} 
	 \label{b1} 
	\end{align}	
	For every $k\in\{1,\dots,\eta\}$, $\nu\in\{1,\dots,n_k\}$ and output trajectory $s_{\tilde{\rho}}({\bm\cdot};\tilde{x}_0,u)$ of trajectory $\tilde{\rho}=\{(q^{i},\xi_{q^{i}}(t;\tilde{x}^{0}_{q^i},u),T^{i})\}_{i=0}^{N_{q}}$ (where $\tilde{x}^{0}_{q^0}=\tilde{x}^{\ast 0}_{q^0}=\tilde{x}_0$) of the switched stochastic control system, if $\sup_{0\leq t\leq T^i}\phi_{q^i}(\xi_{q^i}^{\ast}(t;\tilde{x}^{\ast0}_{q^i},u), \xi_{q^i}(t;\tilde{x}^0_{q^i},u))<\hat{\gamma}$, then $\xi_{q^i}(t;\tilde{x}^0_{q^i},u)\in B_{q^i}(\xi_{q^i}^{\ast}(t;\tilde{x}^{\ast 0}_{q^i},u),\hat{\gamma}e^{-\mu_{q^i} t/2})$, we have
	\begin{align}\nonumber
	\begin{split}
	& (\xi_{q^i}(t;\tilde{x}^0_{q^i},u)-\xi_{q^i}^{\ast}(t;\tilde{x}^{\ast 0}_{q^i},u))^{T}a_{k,\nu}a_{k,\nu}^{T}(z_{k,\nu}^{i})^2(\xi_{q^i}(t;\tilde{x}^0_{q^i},u)\\
	&-\xi_{q^i}^{\ast}(t;\tilde{x}^{\ast 0}_{q^i},u))\le (\xi_{q^i}(t;\tilde{x}^0_{q^i},u)-\xi_{q^i}^{\ast}(t;\tilde{x}^{\ast 0}_{q^i},u))^{T}M_{q^i}\\
	&(\xi_{q^i}(t;\tilde{x}^0_{q^i},u)-\xi_{q^i}^{\ast}(t;\tilde{x}^{\ast 0}_{q^i},u))
		\end{split}  
	\end{align}
		\begin{align}
	\begin{split}
	&=\phi_{q^i}(\xi_{q^i}(t;\tilde{x}^0_{q^i},u), \xi_{q^i}^{\ast}(t;\tilde{x}^{\ast 0}_{q^i},u))e^{-\mu_{q^i} t}\le\hat{\gamma}e^{-\mu_{q^i} t}.
	\end{split}  
	\end{align}
	Therefore, we have $\norm{a_{k,\nu}^{T}(\xi_{q^i}(t;\tilde{x}^0_{q^i},u)-\xi_{q^i}^{\ast}(t;\tilde{x}^{\ast 0}_{q^i},u))}\le\sqrt{\hat{\gamma}}e^{-\mu_{q^i} t/2}/z^i_{k,\nu}$, thus 
	\begin{align}
	\begin{split}
	& -\sqrt{\hat{\gamma}}e^{-\mu_{q^i} t/2}/z_{k,\nu}\le a_{k,\nu}^{T}(\xi_{q^i}(t;\tilde{x}^0_{q^i},u)-\xi_{q^i}^{\ast}(t;\tilde{x}^{\ast 0}_{q^i},u))\\&\le\sqrt{\hat{\gamma}}e^{-\mu_{q^i} t/2}/z^i_{k,\nu}.  
	\label{b2}
	\end{split} 
	\end{align}
	From (\ref{b1}) and (\ref{b2}), we have
	\begin{align}
	\begin{split}
	& -(\sqrt{\hat{\gamma}}+\sqrt{r_{q^i}})e^{-\mu_{q^i} t/2}/z^i_{k,\nu}\le a_{k,\nu}^{T}(\xi_{q^i}(t;\tilde{x}^0_{q^i},u)-\\&\xi_{q^i}^{\ast}(t;x^{\ast 0}_{q^i},u))\le(\sqrt{\hat{\gamma}}+\sqrt{r_{q^i}})e^{-\mu_{q^i} t/2}/z_{k,\nu}. 
	\end{split} 
	\label{b3}
	\end{align}		 	 

If $\left[\left[\varphi_{\hat{\delta}}\right]\right](s_{\rho^\ast}({\bm\cdot};x^{\ast}_0,u), 0)\ge 0$, where $\varphi_{\hat{\delta}}$ is the $\hat{\delta}_{k,\nu}$-robust modified formula of $\varphi$, $\hat{\delta}^i_{k,\nu}=(\sqrt{\hat{\gamma}}+\sqrt{r_{q^i}})/z^i_{k,\nu}$, then for every $k\in\{1,\dots,\eta\}$, $\nu\in\{1,\dots,n_k\}$, $i\in\{1,2,\dots,N_q\}$ (resp. $i=0$), and for any $t$ such that $t+\sum\limits_{j=1}^{i-1}T^j\ge\tau_k$ (resp. $t\ge\tau_k$ when $i=0$), we have $a_{k,\nu}^{T}\xi_{q^i}^{\ast}(t;x^{\ast}_0,u)+c_{k,\nu}^{T}u<b_{k,\nu}-(\sqrt{\hat{\gamma}}+\sqrt{r_{q^i}})e^{-\mu_{q^i} t/2}/z^i_{k,\nu}$. In such conditions, for any $\tilde{x}_0\in B_{q^0}(x^{\ast}_0,r_{q^0})$ (thus $\tilde{x}^{\ast0}_{q^i}\in B_{q^i}(x^{\ast 0}_{q^i},r_{q^i})$), if $\xi_{q^i}(t;\tilde{x}^0_{q^i},u)\in B_{q^i}(\xi_{q^i}^{\ast}(t;\tilde{x}^{\ast 0}_{q^i},u),\hat{\gamma}e^{-\mu_{q^i} t/2})$, we have
	\begin{align}
	\begin{split}\nonumber 
	& a_{k,\nu}^{T}\xi_{q^i}(t;\tilde{x}^0_{q^i},u)+c_{k,\nu}^{T}u< a_{k,\nu}^{T}\xi_{q^i}^{\ast}(t;x^{\ast 0}_{q^i},u)+c_{k,\nu}^{T}u+(\sqrt{\hat{\gamma}}\\&+\sqrt{r_{q^i}})e^{-\mu_{q^i} t/2}/z^i_{k,\nu}<b_{k,\nu}.
	\end{split}
	\end{align}     
Therefore, from the above analysis and (\ref{prob}), for any $\tilde{x}_0\in B_{q^0}(x^{\ast}_0,r_{q^0})$ we have ($0\le t\le T^i$ in the following notations)
\[ 
	\begin{split}
	& P\{\left[\left[\varphi\right]\right](s_{\tilde{\rho}}({\bm\cdot};\tilde{x}_0,u), 0)\ge 0~\vert~\left[\left[\varphi_{\hat{\delta}}\right]\right](s_{\rho^\ast}({\bm\cdot};x^{\ast}_0,u), 0)\ge 0\}\\
	& \ge P\{\forall k, \forall \nu, \forall i,\forall t~\textrm{such~that}~t+\sum\limits_{j=1}^{i-1}T^j\ge\tau_k (\textrm{resp.}~ t\ge\tau_k ~\\&\textrm{when}~ i=0), a_{k,\nu}^{T}\xi_{q^i}(t;\tilde{x}^0_{q^i},u)+c_{k,\nu}^{T}u<b_{k,\nu}~\vert~\left[\left[\varphi_{\hat{\delta}}\right]\right]\\&~~ (s_{\rho^\ast}({\bm\cdot};x^{\ast}_0,u),0)\ge 0\}
	\\&\ge P\{\forall k, \forall \nu, \forall i,\forall t, \norm{a_{k,\nu}^{T}(\xi_{q^i}(t;\tilde{x}^0_{q^i},u)-\xi_{q^i}^{\ast}(t;x^{\ast 0}_{q^i},u))}<\\&~~~(\sqrt{\hat{\gamma}}+\sqrt{r_{q^i}})e^{-\mu_{q^i} t/2}/z_{k,\nu}~\vert\left[\left[\varphi_{\hat{\delta}}\right]\right](s_{\rho^\ast}({\bm\cdot};x^{\ast}_0,u), 0)\ge 0\}
	\\&= P\{\forall k, \forall \nu,\forall i,\forall t,\norm{a_{k,\nu}^{T}(\xi_{q^i}(t;\tilde{x}^0_{q^i},u)-\xi_{q^i}^{\ast}(t;x^{\ast 0}_{q^i},u))}<\\&~~~(\sqrt{\hat{\gamma}}+\sqrt{r_{q^i}})e^{-\mu_{q^i} t/2}/z_{k,\nu}\} 
	\\&\ge P\{\forall k, \forall \nu,\forall i,\forall t,\norm{a_{k,\nu}^{T}(\xi_{q^i}(t;\tilde{x}^0_{q^i},u)-\xi_{q^i}^{\ast}(t;\tilde{x}^{\ast 0}_{q^i},u))}<\\&~~~\sqrt{\hat{\gamma}}e^{-\mu_{q^i} t/2}/z_{k,\nu}\}
	\\&\ge P\{\sup_{0\leq t\leq T^0}\phi_{q^0}(\xi_{q^0}^{\ast}(t;\tilde{x}^{\ast0}_{q^0},u), \xi_{q^0}(t;\tilde{x}^0_{q^0},u))<\hat{\gamma}\}\times\dots\\
	& P\{\sup_{0\leq t\leq T^{N_q}}\phi_{q^{N_q}}(\xi_{q^{N_q}}^{\ast}(t;\tilde{x}^{\ast0}_{q^{N_q}},u), \xi_{q^{N_q}}(t;\tilde{x}^0_{q^{N_q}},u))<\hat{\gamma}\}	\\
	&\ge(1-\frac{\alpha_{q^0} T^0}{\hat{\gamma}})\times(1-\frac{\alpha_{q^1} T^1}{\hat{\gamma}})\times\dots
	\times(1-\frac{\alpha_{q^{N_q}}T^{N_q}}{\hat{\gamma}})	\\
 	&\stackrel{(a)}{\ge} 1-\frac{\alpha_{q^0}T^0+\alpha_{q^1}T^1+\dots \alpha_{q^{N_q}}T^{N_q}}{\hat{\gamma}}	\\
	&\ge 1-\frac{(\max\limits_{i}\alpha_{q^i})\cdot (T^0+T^1+\dots T^{N_q})}{\hat{\gamma}}\\
	&=1-\frac{(\max\limits_{i}\alpha_{q^i})\cdot T_{\textrm{end}}}{\hat{\gamma}}=1-\epsilon.    
	\end{split}  
\]           
The inequality $(a)$ follows from the fact that $(1-c_1)(1-c_2)\dots(1-c_n)\ge 1-(c_1+c_2+\dots+c_n)$ when $0\le c_i\le 1$ $(i=1,2,\dots,n)$, which can be easily proven by induction.

\bibliographystyle{IEEEtran}
\bibliography{zhepowerref}
\end{document}